\documentclass[sts]{imsart}

\RequirePackage{amsthm,amsmath,amsfonts,amssymb}
\RequirePackage[numbers,sort&compress]{natbib}
\RequirePackage[colorlinks,citecolor=blue,urlcolor=blue]{hyperref}
\RequirePackage{graphicx}

\usepackage{stmaryrd} 
\usepackage{makecell}
\usepackage{xcolor}
\usepackage{bm} 

\startlocaldefs
\theoremstyle{plain}

\newtheorem{theorem}{Theorem}[section]
\newtheorem{lemma}[theorem]{Lemma}
\theoremstyle{definition}

\newcommand{\mT}{\mathcal{T}}
\newcommand{\mY}{\mathcal{Y}}
\newcommand{\qmq}[1]{\quad\mbox{#1}\quad}
\newcommand{\mq}[1]{\mbox{#1}\quad}

\endlocaldefs

\begin{document}

\begin{frontmatter}
\title{Shortest Fixed-Width Confidence Intervals for a Bounded Parameter: The Push Algorithm}
\runtitle{Shortest Fixed-Width Confidence Intervals}

\begin{aug}
\author[A]{\fnms{Jay}~\snm{Bartroff}\ead[label=e1]{bartroff@austin.utexas.edu}}
\and
\author[B]{\fnms{Asmit}~\snm{Chakraborty}\ead[label=e2]{ac83832@my.utexas.edu}}

\address[A]{Jay Bartroff is Professor, Department of Statistics and Data Sciences,
University of Texas at Austin, Austin, USA\printead[presep={\ }]{e1}.}

\address[B]{Asmit Chakraborty is Graduate Student, Department of Statistics and Data Sciences,
University of Texas at Austin, Austin, USA\printead[presep={\ }]{e2}.}

\end{aug}

\begin{abstract}
We present a method for computing optimal fixed-width confidence intervals for a single, bounded parameter, extending a method for the binomial due to Asparaouhov and Lorden, who called it the Push algorithm. The method produces the shortest possible non-decreasing confidence interval for a given confidence level, and if the Push interval does not exist for a given width and level, then no such interval exists. The method applies to any bounded parameter that is discrete, or is continuous for which interval probabilities are quasiconcave.  We demonstrate the method on the binomial, hypergeometric, and normal distributions with our available R package. In each of these distributions the proposed method outperforms the standard ones, and in the latter case even improves upon the $z$-interval. We apply the proposed method to World Health Organization (WHO) data on tobacco use.\end{abstract}

\begin{keyword}
\kwd{binomial distribution}
\kwd{confidence intervals}
\kwd{exact methods}
\kwd{quasiconcavity}
\kwd{total positivity}
\end{keyword}

\end{frontmatter}

\section{Introduction}
In a number of statistical applications, the maximum width of a confidence interval (or its margin of error) resulting from a study is specified \textit{a priori} in guidance or regulation.  Examples include
the Trends in International Mathematics and Science Study (TIMSS) study, a large-scale international assessment of students' mathematics and science knowledge conducted every 4 years in which confidence interval widths are prescribed in terms of student test score points \citep{Siegel24}, and numerous surveys regularly conducted by the U.S.\ Department of Agriculture \citep[see][Appendix~D]{Gearing21,Magness21}, in which confidence interval widths are similarly specified. When the inferential criterion is a prescribed maximum width, restricting attention to fixed-width intervals entails no loss in achievable coverage: Any interval of width at most~$w$ can be enlarged to one of width exactly~$w$ without decreasing coverage. Thus fixed-width intervals provide a natural formulation of the problem of finding the smallest attainable uniform bound on interval width in settings such as these.

The study of fixed-width confidence intervals began with Stein’s~\citeyearpar{Stein45} two-stage procedure for the normal mean with unknown variance, which uses a pilot study to estimate the variance and then chooses the total sample size  just large enough to ensure the confidence interval has the pre-determined width.  Since \citet{Stein45}, the development of theory for fixed-width intervals has primarily been in the multistage and sequential domain, and few authors consider the finite-sample optimality developed here.  Sequential methods were proposed by \citet{Anscombe53}  for the normal mean, and by \citet{Chow65} who showed their  method achieved  asymptotic coverage and first-order optimality in expected sample size \citep[see also][]{Starr66,Simons68,Mukhopadhyay80,Woodroofe86}. Refinements by \citet{Serfling76} and \citet{Mukhopadhyay96} addressed finite-sample properties and second-order accuracy, and \citet{Srivastava67} explicitly describes confidence regions with ``prescribed maximum width and prescribed coverage probability.'' A sequential method for the binomial success probability with conservative coverage guarantees was proposed by \citet{Frey10}. The textbooks \citet{Mukhopadhyay94} and \citet{Ghosh11} give comprehensive treatments of these sequential and multistage approaches.

Asparouhov's \citeyearpar{Asparouhov00} PhD thesis includes a computational method for optimal fixed-width confidence intervals for the binomial success probability, which he credited to an unpublished manuscript by the thesis advisor, \citet{Lorden00}.  They called this method ``Push'' because it chooses the lower endpoint of the fixed-width interval to increase as rapidly as possible as a function of the statistic, which they show is both necessary and sufficient for optimality. The method is recursive but generally quick, and therefore useful in applications calling for fixed-width intervals like those above. But the method appears to be essentially unknown to the statistics community due to its unpublished status, lack of software, and also perhaps because \citet{Asparouhov00} describes it in such generality to include sequential sampling; here we focus on the fixed-sample setting.  The goals of this paper are to bring this method to a wider audience through exposition and demonstration of the method and our available R package \citep{Chakraborty25}, and to generalize it from the binomial to a wider class of distributions including the binomial, hypergeometric, and normal distributions, which we demonstrate below.

Three key ingredients in the Push algorithm are boundedness of the parameter, randomization of the statistic (if discrete), and discretization of the parameter space (if continuous). Although these last two may sound contradictory, randomization of a discrete statistic allows its quantiles  to be uniquely defined, and discretizing a continuous parameter space allows recursive computation of the optimal intervals. The grid chosen to discretize a continuous parameter may be taken to be arbitrarily fine, allowing arbitrary efficiency of the method relative to truly continuous confidence intervals.  

In Section~\ref{sec:gen} we formulate the general setup for the Push intervals and prove their coverage and optimality in Theorem~\ref{thm:main}. In Sections~\ref{sec:binom.intervals}, \ref{sec:hyper}, and \ref{sec:norm.mean} we apply the general method to the binomial, hypergeometric, and normal mean problems, respectively, including simulation studies of their performance and comparisons with the standard methods. In Section~\ref{sec:WHO.dat} we apply the Push intervals to World Health Organization (WHO) data on tobacco use, before concluding in Section~\ref{sec:Concl}.  Throughout  we use $\vee$ and $\wedge$ for maximum and minimum, respectively, and $\llbracket x \rrbracket$ for  the integer closest to $x$, i.e., conventional rounding.

\section{General formulation}\label{sec:gen}

\subsection{Set up and Push algorithm definition}\label{sec:setup}
Let $Y$ be a continuous random variable taking values in interval~$\mY\subseteq\mathbb{R}$ with density from a family~$\{f_\theta:\; \theta\in\Theta\subseteq\mathbb{R}\}$ whose parameter space contains the interval~$[\underline{\theta},\overline{\theta}]\subseteq\Theta$, and whose c.d.f.\ is continuous and strictly increasing.
For chosen $m\ge 1$ we fix the grid  
\begin{equation}\label{gen.grid.def}
\theta_k=\underline{\theta}+(\overline{\theta}-\underline{\theta})k/m,\quad k=0,1,\ldots, m,
\end{equation} 
so that $\theta_0=\underline{\theta}$ and $\theta_m=\overline{\theta}$. The grid size is $(\overline{\theta}-\underline{\theta})/m$ and we will denote the desired width of confidence intervals for $\theta$ by 
\begin{equation}\label{w.def}
w=(\overline{\theta}-\underline{\theta})r/m  
\end{equation} for some $r\in\{1,\ldots,m\}$. If $\theta$ is a continuous parameter, then the $\theta_k$ are used as a discretization of $[\underline{\theta},\overline{\theta}]$ useful in computing the Push algorithm confidence intervals, defined below, and $m$ is chosen for desired accuracy.  If $\theta$ is discrete, equally spaced,  and bounded, then $\underline{\theta}$, $\overline{\theta}$, and $m$ are chosen so that the $\theta_k$ are the parameter values themselves.  For example, for  the hypergeometric distribution considered in Section~\ref{sec:hyper} in which $\theta$ is the number of successes in a population of size $N$, we take $\underline{\theta}=0$ and $\overline{\theta}=m=N$ so that $\theta_k=0,1,\ldots, N$.

Throughout we let $\gamma\in(0,1)$ denote the desired confidence level. We consider intervals~$[L(Y), R(Y)]$ that satisfy 
\begin{equation}\label{rest.conf.lev}
P_\theta(\theta\in [L(Y), R(Y)])\ge \gamma\qmq{for all} \theta\in [\underline{\theta}, \overline{\theta}].
\end{equation}
For discrete $\theta$, we use the set $[\underline{\theta}, \overline{\theta}]$ in \eqref{rest.conf.lev} to denote $\{\underline{\theta}=\theta_0,\theta_1,\ldots,\theta_m=\overline{\theta}\}$. For continuous $\theta$, $[\underline{\theta}, \overline{\theta}]$ denotes the usual interval. In some applications  of the proposed method (such as for the binomial and hypergeometric, below), $[\underline{\theta}, \overline{\theta}]$ will be the entire parameter space and thus \eqref{rest.conf.lev} is the usual definition of the interval~$[L(Y), R(Y)]$ having confidence level~$\gamma$. In settings where the true parameter space is unbounded, such as the normal mean problem considered in Section~\ref{sec:norm.mean}, $[\underline{\theta}, \overline{\theta}]$ is a subset of the parameter space, perhaps motivated by prior information about $\theta$. In that case,  \eqref{rest.conf.lev} corresponds with the \textit{restricted} or \textit{conditional} confidence level which has been considered  in the statistics literature in other contexts \citep[e.g.,][]{Farchione08,Kabaila24,Mandelkern02,Wang08,Zhang03}.   

In considering fixed-width confidence intervals of width~\eqref{w.def}, a simplification is that we consider only $\{\theta_k\}$-valued intervals, and thus of the form 
\begin{multline*}
[L(Y), R(Y)] = [L(Y), L(Y)+w]\\
\mq{where} w=(\overline{\theta}-\underline{\theta})r/m,
\end{multline*}
with $L(y)$ restricted to taking values in $\{\theta_k\}$.  To describe the values $R(y)$ may take we thus extend \eqref{gen.grid.def} beyond $k=m$ by setting $\theta_k= (\underline{\theta}+(\overline{\theta}-\underline{\theta})k/m) \wedge\sup(\Theta)$ for $k>m$. This reduction to grid-constrained intervals is essentially without loss of generality since
\citet{Asparouhov00} points out that the difference between shortest-width continuous and grid-constrained intervals in the continuous-$\theta$ case is at most $2/m$, which can be made arbitrarily small by the choice of $m$. And nothing is lost in the discrete $\theta$ case where $\Theta=\{\theta_0,\ldots, \theta_m\}$. In our numerical examples below we take $m=10^5$, making $2/m$ smaller than is typically recorded in most applications. 

Let $F_k(y)$ denote the c.d.f.\ of $Y$ when $\theta=\theta_k$, $F_k^{-1}(\beta)$ its quantile function for $\beta\in[0,1]$, and extend the domain of $F_k^{-1}$ by setting
\begin{equation}\label{bin.qF.ext}
F_k^{-1}(\beta)=\infty\qmq{for}\beta>1.  
\end{equation}  
The Push algorithm $\gamma$-confidence interval of fixed-width $w=(\overline{\theta}-\underline{\theta})r/m$ is 
\begin{equation}\label{gen.push.def}
[L^*(y),R^*(y)]=[\theta_k,\theta_{k+r}]\qmq{for} y_k\le y<y_{k+1}
\end{equation} where the $y_k$, the \textit{generalized inverses} of $L^*$, are defined slightly differently for whether $\theta$ is continuous or discrete.  In both cases set
\begin{equation}\label{gen.y.init}
 y_{-r}=y_{-r+1}=\ldots =  y_0=\inf(\mY), \quad  y_{m+1}=\infty.
\end{equation} If $\theta$ is continuous, define the remaining $y_k$ recursively as
\begin{multline}
\label{gen.recur.cont}
y_k=\\
y_{k-1}\vee F_{k-1}^{-1}(\gamma+F_{k-1}(y_{k-r})) \vee F_{k}^{-1}(\gamma+F_{k}(y_{k-r})),\\
k=1,\ldots, m.
\end{multline}
If $\theta$ is discrete, define them as
\begin{equation}\label{gen.recur.disc}
y_k=y_{k-1}\vee F_{k-1}^{-1}(\gamma+F_{k-1}(y_{k-r-1})),\quad k=1,\ldots, m.
\end{equation}

The recursions~\eqref{gen.recur.cont} and \eqref{gen.recur.disc} can always be completed to $k=m$, however because of \eqref{bin.qF.ext} it may be that $y_k= \ldots=y_m=\infty$ for some $k\le m$. This happens in the continuous case when
\begin{equation}\label{push.ex.cont}
\gamma+F_{k-1}(y_{k-r})>1\qmq{or} \gamma+F_{k}(y_{k-r})>1,
\end{equation} 
and in the discrete case when 
\begin{equation}
\gamma+F_{k-1}(y_{k-r-1})>1.
\end{equation} When the Push recursion satisfies $y_m<\infty$, we say that the Push interval \textit{exists} for given $\gamma$ and $r$.

\subsection{Coverage and optimality of the Push algorithm}

Our main results on the coverage and optimality of the Push algorithm are in Theorem~\ref{thm:main}.  In the continuous-parameter case, we will require that the statistic's interval probabilities have the following property: Recall that a function $g(\theta)$ is \textit{quasiconcave} if, for any constant $c$, the upper level set~$\{\theta: g(\theta)\ge c\}$ is convex \citep[see][p.~95]{Boyd04}. Since here we consider real-valued $\theta$, in our setting this is equivalent to upper level sets being intervals. Sufficient conditions on $f_\theta(y)$ for this quasiconcavity to hold are discussed after the theorem.

\begin{theorem}\label{thm:main} Let $Y$ be a continuous random variable with density~$f_\theta(y)$, $\theta\in\Theta \subseteq\mathbb{R}$, whose c.d.f.\ is continuous and strictly increasing, and  either:
\begin{enumerate}
\item[(C)] $\theta$ is continuous with $[\underline{\theta},\overline{\theta}]\subseteq\Theta$ and, for any interval $I\subseteq \mY$, $P_\theta(Y\in I)$  is continuous and quasiconcave in $\theta$; \\
or
\item[(D)] $\theta$ is discrete with $\Theta$ given by \eqref{gen.grid.def}.  
\end{enumerate}
 If there exists a $\{\theta_k\}$-constrained interval~$[L(Y),R(Y)]$ of fixed-width $(\overline{\theta}-\underline{\theta})r/m$ satisfying \eqref{rest.conf.lev}
whose endpoints are non-decreasing in $Y$, then  the Push algorithm interval~$[L^*(Y),R^*(Y)]$ (given by \eqref{gen.recur.cont} or \eqref{gen.recur.disc}, respectively) exists,  satisfies \eqref{rest.conf.lev}, and $L^*(y)\ge L(y)$ and $R^*(y)\ge R(y)$ for all $y$.
\end{theorem}

The theorem says that if the Push interval does not exist for a given width and desired confidence level, then no such increasing confidence interval exists of that width. For discrete statistics~$X$, the continuous random variable~$Y$ will be taken to be a smoothed (or ``randomized'') version, which in the applications to the binomial and hypergeometric below we take to be $Y=X+U$ where $U$ is a uniform random variable independent of $X$.  

In the continuous case~(C), a sufficient condition for the required quasiconcavity is that $f_\theta(y)$ be totally positive of order $3$, defined next, and a large class of distributions satisfying this are one-parameter exponential families $f_\theta(y)=b(y)\exp[\eta(\theta) T(y)-\psi(\theta)]$ with non-decreasing natural parameter~$\eta(\theta)$ and sufficient statistic~$T(y)$.  A function \(K(\theta,y)\) on two linearly ordered sets is \textit{totally positive of order $n$ (TP$_n$)} if,  for every $1\leq k\leq n$, $\theta_1<\ldots<\theta_k$, and $y_1<\ldots<y_k$, we have
$\det\left[K(\theta_i,y_j)\right]_{i,j=1}^k\ge 0$  \citep[see][]{Karlin68}. These classes are nested in the sense that  TP$_1\supseteq$ TP$_2\supseteq\cdots$, and we say that $K$ is TP$_\infty$ if it is TP$_n$ for all $n$.  Lemma~\ref{lem:grid.cov}, Part~(ii), verifies that $K(\theta,y)=f_\theta(y)$ being TP$_3$ is sufficient for the interval probabilities~$P_\theta(Y\in I)$ to be quasiconcave. One-parameter exponential families  with non-decreasing natural parameters and sufficient statistics are TP$_\infty$  \citep[][p.~19]{Karlin68}, including the normal mean and randomized binomial families considered below, and many other commonly-encountered one-parameter distributions. One-parameter exponential families in which both the natural parameter and sufficient statistic are monotone, but not necessarily nondecreasing, can also be made to satisfy this condition after reversing the order of the parameter or sample space as needed.  An example of a one-parameter exponential family that \emph{fails} this requirement is the spherical distribution $f_\theta(y)\propto \exp(\theta \cos^2 y)$ \citep[see][]{Watson66} due to its non-monotonic sufficient statistic $\cos^2(y)$. A non-exponential family which is TP$_2$ but not TP$_3$ is $f_\theta(y)\propto \sqrt{1+\theta y}$ for $\theta, y\in(-1,1)$,  which can be shown to have non-interval upper level sets of $P_\theta(Y\in I)$.

\begin{proof}[Proof of Theorem~\ref{thm:main}] Let $[L,R]$ be as in the theorem, and let $x_k=\inf\{y:\; L(y)\ge\theta_k\}$ denote the generalized inverses of $L$ so that, if $x_k<y<x_{k+1}$ then $[L(y),R(y)]=[\theta_k,\theta_{k+r}]$. 

First suppose that the continuous case~(C) holds, in which Lemma~\ref{lem:grid.cov}, Part~(i), shows that \eqref{rest.conf.lev} is equivalent to the grid coverage condition \eqref{gen.grid.cov}, which we can write as
\begin{multline}
\label{grid.cov.L}
 \gamma\le P_\theta(L(Y)\le \theta_{k-1}<\theta_k\le R(Y)) \\
 = P_\theta(x_{k-r}< Y<  x_k)\qmq{for}\theta=\theta_{k-1}, \theta_k.   
\end{multline}
These two inequalities yield
\begin{equation*}
x_k\ge F_j^{-1}(\gamma+F_j(x_{k-r}))\qmq{for} j=k-1,k,
\end{equation*} thus
\begin{equation}\label{cont.recur.ineq}
x_k\ge x_{k-1}\vee F_{k-1}^{-1}(\gamma+F_{k-1}(x_{k-r})) \vee F_{k}^{-1}(\gamma+F_{k}(x_{k-r})),
\end{equation} since $L$ is increasing. Comparing this with \eqref{gen.recur.cont}, it follows by induction that $x_k\ge y_k$ for all $k$, and thus $L$ existing implies that $\infty>x_m\ge y_m$, which implies that $L^*$ exists. Since \eqref{grid.cov.L}-\eqref{cont.recur.ineq} hold with $y_k$ in place of $x_k$, the grid coverage condition~\eqref{gen.grid.cov} holds for $[L^*,R^*]$, for which \eqref{rest.conf.lev} holds as well by Lemma~\ref{lem:grid.cov}.

If instead the discrete case~(D) holds, with $w$ equal to \eqref{w.def}, the coverage of $[L,R]$  implies that
\begin{multline*}
\gamma\le P_{\theta_{k-1}}(L(Y)\le \theta_{k-1}\le L(Y)+w) \\
= P_{\theta_{k-1}}(x_{k-r-1}< Y< x_k),
\end{multline*}
thus 
\begin{equation*}
x_k\ge x_{k-1}\vee F_{k-1}^{-1}(\gamma+F_{k-1}(x_{k-r-1})).
\end{equation*} By \eqref{gen.recur.disc} and an inductive argument similar to the continuous case, we have $x_k\ge y_k$ from which it follows that $[L^*,R^*]$ exists if $[L,R]$ does, and coverage probability~\eqref{rest.conf.lev} holds for the former since these last two inequalities hold with $y_k$ in place of $x_k$.
\end{proof}

The next lemma contains two auxiliary results related to Theorem~\ref{thm:main}, and is proved in the \hyperref[appn]{Appendix}. Part~(i) was the key in the continuous $\theta$ case~(C) of the theorem, showing that under those assumptions the coverage probability~\eqref{rest.conf.lev} is equivalent to a grid coverage condition~\eqref{gen.grid.cov} for grid-constrained intervals. Part~(ii) shows that the density $f_\theta(y)$ being TP$_3$ is sufficient for the quasiconcavity assumed in case~(C) of the theorem.

\begin{lemma}\label{lem:grid.cov}
Let $Y$ be a continuous random variable with density~$f_\theta(y)$ with  continuous parameter $\theta\in \Theta \subseteq\mathbb{R}$,  $[\underline{\theta},\overline{\theta}]\subseteq \Theta$, and whose c.d.f.\ is continuous and strictly increasing.
\begin{enumerate}
\item[(i)] If  case~$(C)$ of Theorem~\ref{thm:main} is satisfied, then for any non-decreasing, $\{\theta_k\}$-constrained interval $[L(Y),R(Y)]$ with $L(y)<R(y)$,  the coverage probability condition~\eqref{rest.conf.lev} holds if and only if
\begin{multline}\label{gen.grid.cov}
P_\theta(L(Y)\le \theta_{k-1}<\theta_k\le R(Y))\ge \gamma\\
\mq{for}\theta=\theta_{k-1}, \theta_k,\quad k=1,\ldots,m.
\end{multline}
\item[(ii)] If $(\theta,y)\mapsto f_\theta(y)$ is TP$_3$ on $[\underline{\theta},\overline{\theta}]\times \mY$ then, for any interval $I\subseteq \mY$, $P_\theta(Y\in I)$  is quasiconcave in $\theta$.
\end{enumerate}
\end{lemma}

\section{Binomial confidence intervals} \label{sec:binom.intervals}

Let $S\sim\mbox{Binom}(n,p)$ where the sample size~$n$ is known and the success probability~$p$ is the unknown parameter of interest. To apply the above method for fixed-width confidence intervals for $p$ we take the random variable in Theorem~\ref{thm:main} to be $Y=S+U$ where $U$ is uniformly distributed over the interval $[-1/2,1/2]$, independent of $S$. The discretization \eqref{gen.grid.def} of the parameter space~$\Theta=[\underline{\theta},\overline{\theta}] = [0,1]$ of $p$  is $\{0,1/m,2/m,\ldots, 1\}$ (i.e., $\theta_k=k/m$) and the desired width is denoted by $r/m$. The Push intervals~\eqref{gen.push.def} for $p$ are
\begin{equation}\label{bin.push.def}
[L^*(y),R^*(y)]= \left[\frac{k}{m}, \frac{k+r}{m}\right]\qmq{for} y_k\le y<y_{k+1}
\end{equation} and the  recursion for computing the $y_k$ is \eqref{gen.recur.cont} where $F_k$ is the c.d.f.\ of $Y$ with parameter $p=\theta_k=k/m$, and the initial values in \eqref{gen.y.init} are
\begin{equation*}
 y_{-r}=y_{-r+1}=\ldots =  y_0=-1/2.
\end{equation*} 

In Figure~\ref{fig:binom_intervals}, the binomial Push intervals~$[L^*(y),R^*(y)]$ for $n = 10$, $\gamma = .8$, $w= .318$, and $m = 10^5$ are visually depicted.  We note the intervals' typical asymmetry, truncation at $1$ for larger values of $y$, and slightly curved boundaries between integer values of $y$ due to randomization. These properties are discussed more in Sections~\ref{sec:binom.symm}-\ref{sec:binom.sim}. The R package \citet{Chakraborty25} was used to calculate all Push intervals in this paper, and all figures  were made with \texttt{ggplot2}~\citep{Wickham16}.

\begin{figure}[b]
  \centering
  \includegraphics[width=9cm]{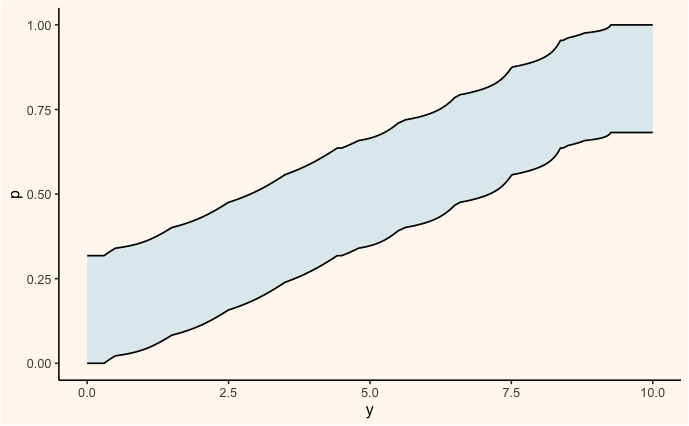}
  \caption{Push intervals~$[L^*(y),R^*(y)]$ (vertical axis, in blue) for binomial $p$ with $n = 10$ trials, nominal confidence level $\gamma = .8$, grid width $m = 10^5$, and width $w=.318$.}
  \label{fig:binom_intervals}
\end{figure}

\subsection{Computational details}
To relate the c.d.f.~$F_k$ of the smoothed random variable~$Y$ utilized in the Push recursion to that of the binomial distribution, let $G_k, g_k$ denote the c.d.f.\ and density, respectively,  of the $\mbox{Binom}(n,p=k/m)$ distribution. Recalling that $\llbracket y\rrbracket$ is $y$ rounded, the c.d.f.~$F_k(y)$ of $Y$ can be written    
\begin{multline}
\label{bin.Y.cdf}
F_k(y) = G_k(\llbracket y\rrbracket-1)+g_k(\llbracket y\rrbracket)(y-\llbracket y\rrbracket+1/2)\\
\mq{for} -1/2\le y \le n+1/2.
\end{multline}
To define the quantile function $F_k^{-1}$, note that $F_k(y)$ is the continuous, piecewise linear function that bisects the ``steps'' of $G_k$.  Given $\beta\in(0,1)$, let  $\underline{\beta}=G_k(s_\beta)$ denote the largest value of $G_k$ that is $\le\beta$.  Then for $g_k(s_\beta+1)> 0$,
\begin{equation}\label{bin.Y.quant}
F_k^{-1}(\beta) = s_\beta+ \frac{\beta-\underline{\beta}}{g_k(s_\beta+1)}+\frac{1}{2}.
\end{equation}

\subsection{Symmetric intervals}\label{sec:binom.symm}
Some confidence intervals~$[L(Y),R(Y)]$ for the binomial parameter $p$ are constructed to be symmetric in the sense of 
\begin{equation}\label{symm.binom}
[L(n-Y),R(n-Y)]=[1-R(Y),1-L(Y)].
\end{equation}
As the example in Figure~\ref{fig:binom_intervals} shows, the Push binomial intervals are in general  \emph{not} symmetric, and Theorem~\ref{thm:main} shows that this asymmetry is necessary for optimality.  Still, some practitioners may require symmetry and it is possible to modify the Push intervals to achieve \eqref{symm.binom} by a ``union with its mirror'' approach which, in general, replaces a binomial confidence interval $[L(Y),R(Y)]$ with
\begin{multline}
\label{union.with.mirror}
[L_{sym}(Y), R_{sym}(Y)]= \\
[L(Y) \wedge (1 - R(n - Y)), R(Y) \vee (1 - L(n-Y))],
\end{multline} which satisfies \eqref{symm.binom}. Note that $$[L(Y), R(Y)] \subseteq [L_{sym}(Y), R_{sym}(Y)]$$ which implies that the confidence level of  \eqref{union.with.mirror} is at least as high as that of $[L(Y), R(Y)]$, but \eqref{union.with.mirror} may be wider than $[L(Y), R(Y)]$.  Thus, to apply this to the Push to achieve symmetric $\gamma$-confidence intervals, we recommend the following steps: 
\begin{enumerate}
\item Find the smallest width $w^*=r^*/m$ for which the Push intervals~$[L^*(Y), R^*(Y)]$ exist for confidence level $\gamma$.
\item Obtain the symmetric intervals $[L_{sym}^*(Y), R_{sym}^*(Y)]$ given by \eqref{union.with.mirror}.
\end{enumerate}
The resulting intervals~$[L_{sym}^*(Y), R_{sym}^*(Y)]$ will be symmetric, have confidence level $\gamma$, but may no longer enjoy the optimality of Theorem~\ref{thm:main} because their resulting width may exceed $w^*$. We investigate the achieved widths of these intervals in our numerical simulations in Section~\ref{sec:binom.sim} and find that in many cases the achieved width of $[L_{sym}^*(Y), R_{sym}^*(Y)]$ is only slightly wider than $w^*$ and that these symmetric intervals still outperform competing symmetric intervals.  Our R package includes an option for finding the minimal width~$w^*$ and producing the symmetric intervals Push using the steps above.

\subsection{Use of randomized intervals}
The binomial Push intervals are a function of $Y=S+U$, and thus may be viewed as a randomized version of intervals based on the binomial statistic~$S$. Although this randomization is necessary for the optimality of Theorem~\ref{thm:main}, in practice one may wish to compute the intervals from the data $S$ alone.  For this the statistician may choose to simply report the $U=0$ version of the intervals, but better choices are to randomize using $S+U$, or to report all the intervals (and their weights) that result from $Y=S+u$ for $u\in [-1/2,1/2]$.  The R package makes this possible by returning a \emph{function}~$y\mapsto [L^*(y), R^*(y)]$ which can be evaluated to achieve any of these options.

\subsection{Simulation examples and comparisons}\label{sec:binom.sim}
Here we demonstrate and compare the binomial Push intervals with competitors in terms of coverage probability and achieved width through numerical simulations.  We first compare with the standard\footnote{Here and throughout, ``standard'' refers to the simple fixed-width interval centered at the usual point estimator, such as $S/n\pm w/2$ for the binomial. We use this terminology only as shorthand for this natural baseline, not to imply uniqueness among fixed-width procedures.} fixed-width~$w$ interval for $p$ based on $S\sim\mbox{Binom}(n,p)$ which has endpoints 
\begin{equation}\label{bin.std.int}
S/n\pm w/2.    
\end{equation}

We note that the right endpoints of \eqref{bin.push.def} will exceed $1$ when $k>m-r$.  In practice, these or any fixed-width~$w$ intervals~$[L(Y),R(Y)]$ with $R(y)\ge w$ can be constrained to $[0,1]$ by replacing them by 
\begin{equation}\label{bin.push.cons}
[L'(Y), R'(Y)]=[(R(Y)\wedge 1)-w, (R(Y)\wedge 1)]   
\end{equation} which only possibly increases the coverage probability, thus does not decrease the confidence level.  This is the approach taken to the standard and Push intervals here, and  the default behavior in our R package. 

We focus on scenarios with relatively small sample sizes $n = 10, 20$ where other intervals based on the normal approximation are known to fail \citep{Brown01,Brown02}. Using resolution $m = 10^5$, we first obtain the minimum widths $w^*$ at which the Push intervals exist at confidence level~$\gamma = .90$ and $.95$. Then, for $k = 1, \dots, m$, we generate $2000$ data sets with true success probability $p = p_k$ under each $n$ from which we estimate the average coverage probabilities and their standard errors of the Push and standard intervals with the same width~$w^*$, which are plotted in Figures~\ref{fig:binom_avg_results} and \ref{fig:binom_avg_results_sym}. For comparison, Figure~\ref{fig:binom_avg_results} includes the unconstrained version of Push (i.e, the version without the modification~\eqref{bin.push.cons}) and Figure~\ref{fig:binom_avg_results_sym} includes the symmetric version of Push, as described in Section~\ref{sec:binom.symm}. Because the number $m = 10^5$  of $p_k$ is very dense, in these figures we plotted the representative subset of points $p \in \{0, .01, \dots, .99, 1\}$ to maintain visual clarity. 

\begin{figure}[b]
  \centering
  \includegraphics[width=9cm]{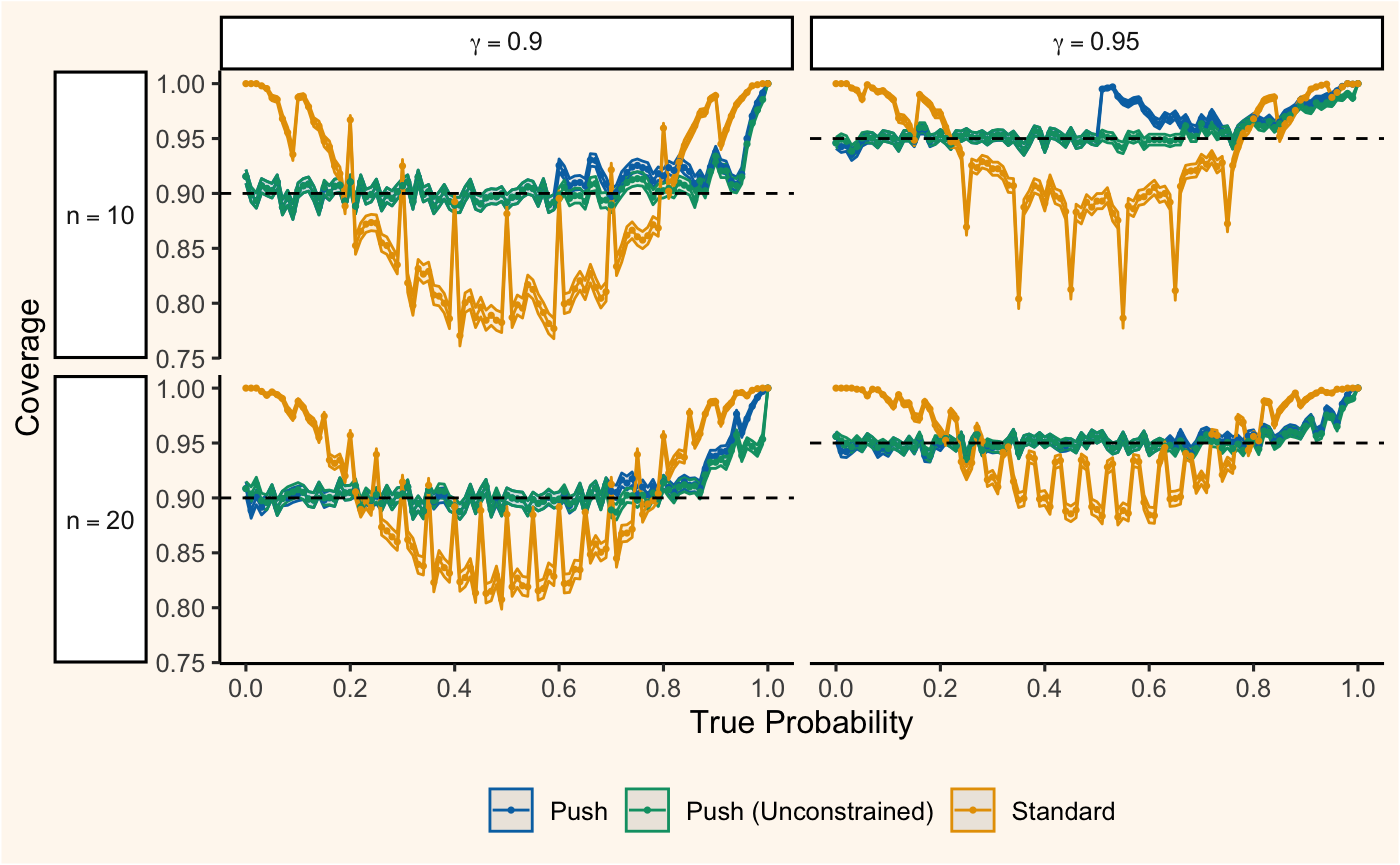}
  \caption{Coverage probability of intervals for binomial~$p$ with $n =10, 20$ trials and nominal confidence levels~$\gamma=.9, .95$.}
  \label{fig:binom_avg_results}
\end{figure}

\begin{figure}[b]
  \centering
  \includegraphics[width=9cm]{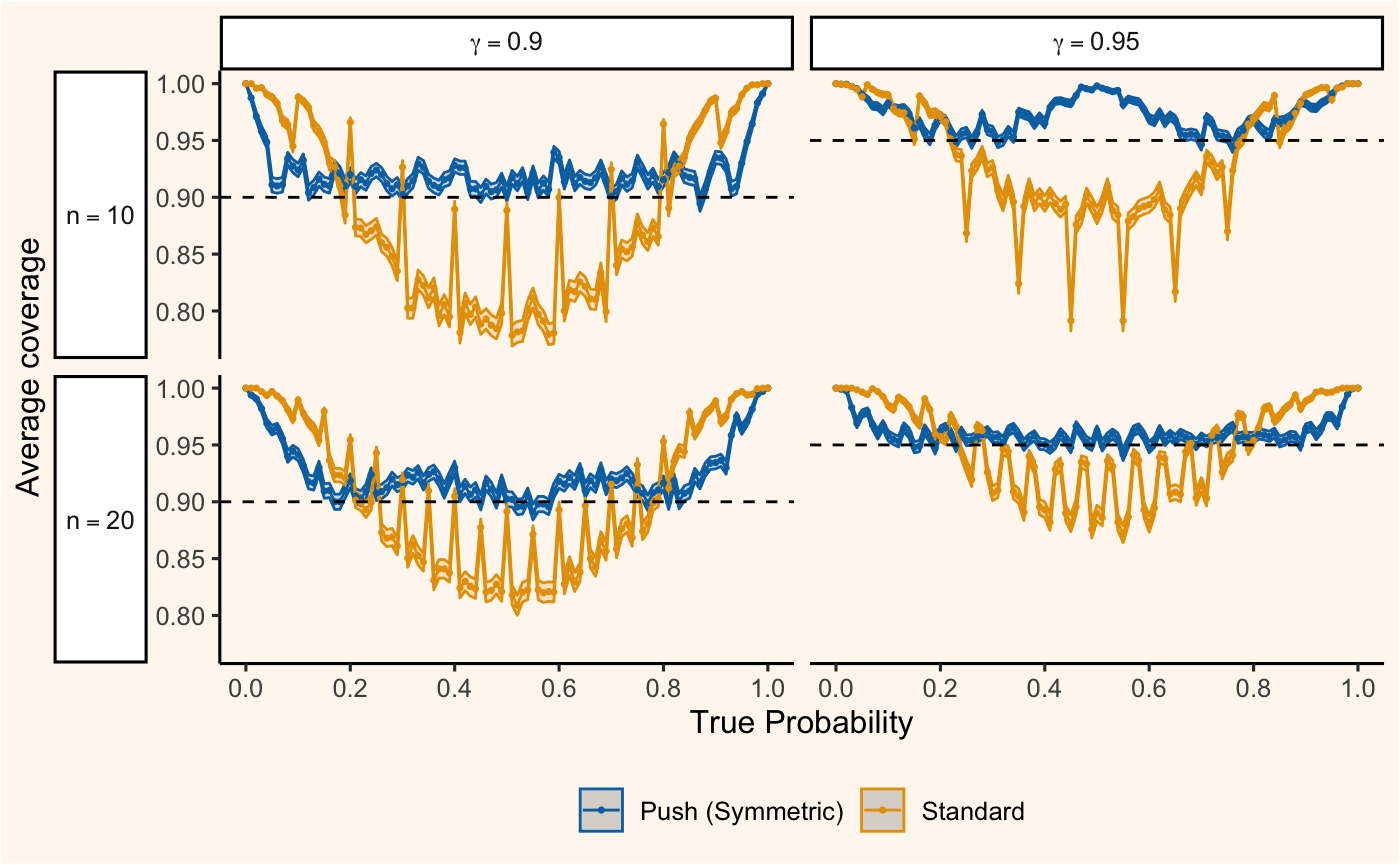}
  \caption{Coverage probability of symmetric intervals for binomial~$p$ with $n =10, 20$ trials and nominal confidence levels~$\gamma=.9, .95$.}
  \label{fig:binom_avg_results_sym}
\end{figure}

Figures~\ref{fig:binom_avg_results} and \ref{fig:binom_avg_results_sym} show the coverage probability as a function of the parameter~$p$.  To also investigate the minimum coverage probability of these methods as a function of achieved width, Figure~\ref{fig:binom_min_results} shows the minimum coverage probability achieved by the Push intervals, the standard intervals of the same width,  and the symmetric version of Push, for the $n=10$ case. For this, we again found the minimum width for which the Push intervals exist at level $\gamma=.7, .8, .9, .95$ which are the widths shown on the $x$-axis of Figure~\ref{fig:binom_min_results}. The asymmetric Push intervals have  minimum coverage probability no smaller than $\gamma$ by Theorem~\ref{thm:main}; the actual minimum coverage probability is slightly larger than, but not visually distinguishable from, $\gamma$ so for this reason the value $\gamma$ is plotted for the asymmetric Push intervals in Figure~\ref{fig:binom_min_results}. The minimum coverage probability of the symmetric Push intervals was obtained by taking the minimum average over the grid of $p$ values as described in the previous paragraph, and for the standard intervals this was done at the union of achieved widths for the asymmetric and symmetric Push intervals. For the Push and standard intervals, the minimum coverage generally occurs for $p$ near $1/2$ and thus the constrained and unconstrained algorithms tend to have the same minimum coverage. Hence, we omit the latter for both methods in the figure.

\begin{figure}[b]
  \centering
  \includegraphics[width=9cm]{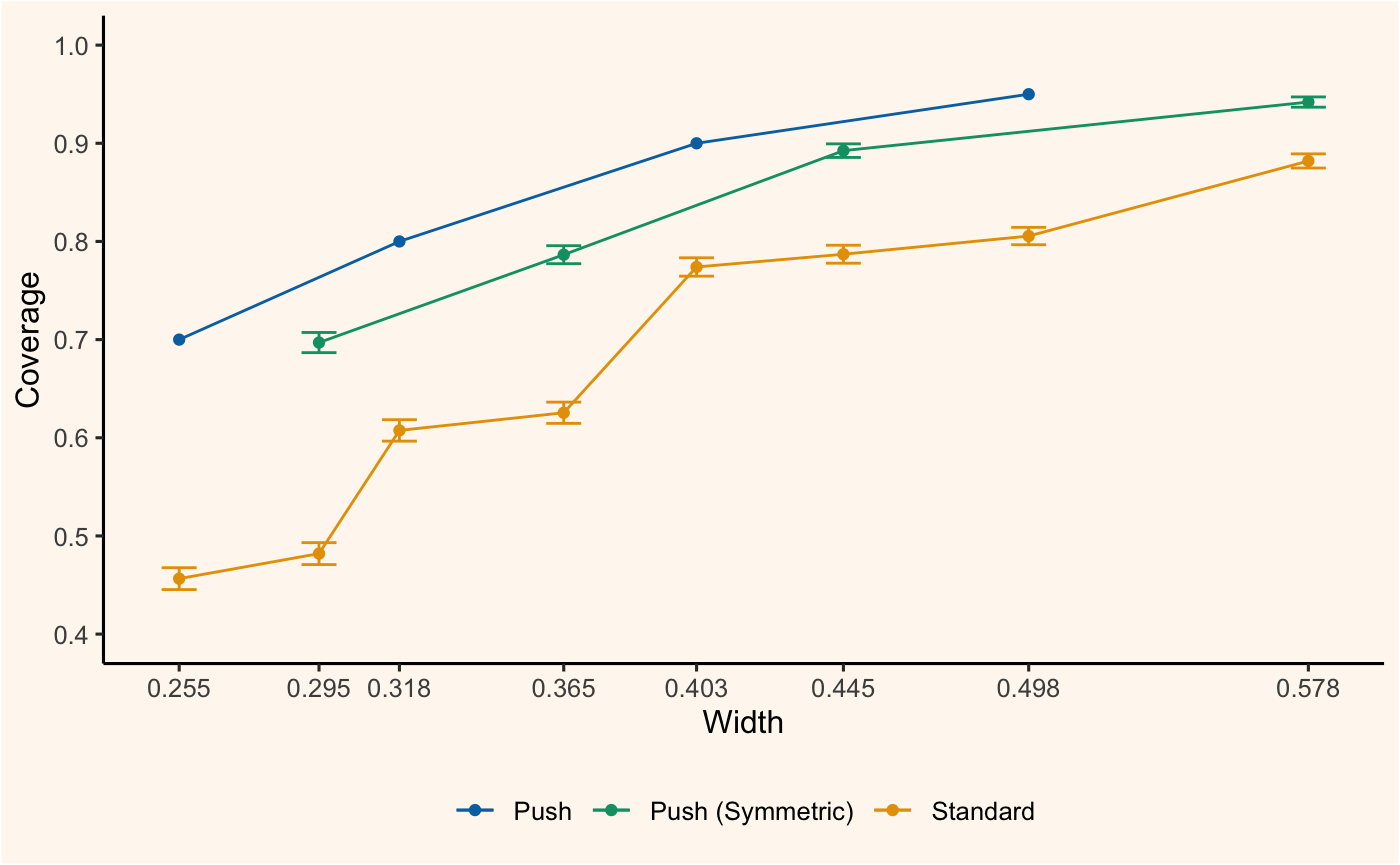}
  \caption{Minimum over $p\in[0,1]$ of coverage probability of binomial Push and standard intervals as a function of their achieved width.}
  \label{fig:binom_min_results}
\end{figure}

In Figures~\ref{fig:binom_avg_results} and \ref{fig:binom_avg_results_sym} the coverage probability of the standard intervals of the same width as the Push intervals  falls well below the nominal level~$\gamma$ achieved by Push intervals for $p$ near the center of $[0,1]$, and approaches $1$ for $p$ near the endpoints of this interval. The coverage probability of the unconstrained Push intervals in Figure~\ref{fig:binom_avg_results} is centered at or just above $\gamma$, except for $p$ near $1$ where it rises due to the ``spillover'' of the right Push endpoints greater than $1$. The coverage probability of the $[0,1]$-constrained Push intervals in Figure~\ref{fig:binom_avg_results} is nearly indistinguishable until $p>1/2$ when this spillover starts occurring and the constraint causes the Push coverage to rise above $\gamma$.

Figure~\ref{fig:binom_avg_results_sym} shows that the favorable performance of the Push intervals in Figure~\ref{fig:binom_avg_results} is not just due to the Push intervals being allowed to be asymmetric. Figure~\ref{fig:binom_avg_results_sym} shows the coverage probability of the standard intervals of the same width, which are symmetric by design, compared with the symmetric Push intervals as described in Section~\ref{sec:binom.symm}. The symmetry constraint causes the coverage probability to rise for $p$ near the endpoints of $[0,1]$ and, in some cases, near the center of that interval, but it remains no smaller than $\gamma$.

A natural question is how much wider the standard intervals in Figures~\ref{fig:binom_avg_results} and \ref{fig:binom_avg_results_sym} would  have to be to maintain  minimum coverage probability $\gamma$ achieved by the Push intervals. Figure \ref{fig:binom_min_results} shows that the needed increase in width would be substantial, even compared with the slightly wider symmetric Push intervals. For example, from the figure we see that the level $\gamma=.7$ Push intervals have width $.255$, but the standard intervals do not achieve that minimum coverage probability until width near $.4$, a more than $35\%$ savings in width by using Push.  The savings at level $\gamma=.8$ is similar, where the smallest width for which the standard intervals achieve minimum coverage is $w=.5$. Here the Push width of $.318$ gives a savings in width of roughly $36\%$.

Because the Push procedure is based on the randomized statistic \(Y=S+U\) with \(U\sim\mathrm{Unif}(-1/2,1/2)\), a natural comparison
is with a randomized version $(S+U)/n \pm w/2$ of the standard centered interval.  For \(n=10\), the minimum widths required by this randomized centered interval to attain confidence levels \(\gamma=.7,.8,.9,.95\) are, to $3$ decimal places, $.340$, $.423$, $.535$, and $.635$, respectively. The corresponding Push widths are $.255$, $.318$, $.403$, and $.497$, yielding reductions in width of about \(25\%,25\%,25\%,\) and \(22\%\). Thus, even relative to a randomized centered benchmark, the Push intervals retain a substantial and fairly consistent advantage in width.

Competing intervals that come closer to the optimal widths of Push are possible with additional numerical optimization.  Extending the approaches of \citet{Wilson27} and \citet{Stevens50},  a randomized interval whose center permits shrinkage toward $1/2$ is
\begin{equation}\label{Wil.shr}
\left[
\frac{Y+a/2}{n+a}-\frac w2,\,
\frac{Y+a/2}{n+a}+\frac w2
\right],
\end{equation}
where \(Y=S+U\), as above, and $a\ge 0$ is a shrinkage parameter. This family contains the randomized interval centered at \(Y/n\) when \(a=0\), while allowing the type of adjusted centering underlying the Wilson interval. For each \(n\) and desired confidence level \(\gamma\), we numerically choose \(a\) and the smallest \(w\) such that the minimum exact binomial coverage is at least \(\gamma\).

\begin{table}[t]
\centering
\caption{Comparison of binomial Push interval widths~$w_P^*$ with the optimized Wilson-type randomized shrinkage interval widths~$w_{shr}^*$ and associated optimized shrinkage parameter~$a^*$}
\label{tab:binom_push_shrink}
\tabcolsep=0pt
\begin{tabular*}{\columnwidth}{@{\extracolsep{\fill}}cccccc@{}}
\hline
$n$ & $\gamma$ & $w_P^*$ & $w_{shr}^*$ &
\makecell[c]{Push width\\reduction (\%)} & $a^\ast$ \\
\hline
10 & .80 & 0.3180 & 0.3211 & 0.968 & 3.161 \\
10 & .90 & 0.4026 & 0.4113 & 2.116 & 3.047 \\
20 & .80 & 0.2379 & 0.2396 & 0.716 & 4.432 \\
20 & .90 & 0.3052 & 0.3063 & 0.364 & 4.346 \\
\hline
\end{tabular*}
\end{table}

Table~\ref{tab:binom_push_shrink} contains the widths of the Push intervals and the optimized intervals~\eqref{Wil.shr} and the associated shrinkage parameter~$a^*$ for $n=10, 20$ and $\gamma=.8,.9$.  The optimized Wilson-type family can come quite close to the optimum, with the Push intervals being no more than about $2.1\%$ shorter in these cases. We emphasize that this remains a particular parametric family of fixed-width procedures, whereas the Push optimality result applies to the full class of nondecreasing grid-constrained intervals considered in Theorem~\ref{thm:main} which guarantees that nothing in the larger admissible class can improve further.


\section{Hypergeometric confidence intervals}\label{sec:hyper}

An approach similar to the binomial above can be taken to achieve optimal fixed-width intervals for the hypergeometric distribution. Let  $X\sim \mbox{Hyper}(\theta,n,N)$ denote a hypergeometric random variable representing the number of successes in a uniform draw, without replacement, of size~$n$ from a population of size~$N$ containing $\theta$ successes. Here we focus on confidence intervals for $\theta$ assuming $n$ and $N$ are known, although a similar approach can be taken for inference about $n$ assuming $\theta$ and $N$ are known. We take the variable in Theorem~\ref{thm:main} to be $Y=X+U$ with $U\sim\mbox{Unif}[-1/2,1/2]$ independent of $X$. 
The parameter space $\Theta=\{0,1,\ldots, N\}$ of  $\theta$ is already naturally discrete so to match this in \eqref{gen.grid.def} we take $\underline{\theta}=0$, $\overline{\theta}=N$, and $m=N$ so $\theta_k=k$. The desired width $w=r\in\{1,\ldots, N\}$ is a nonnegative integer  and the Push intervals~\eqref{gen.push.def} for $\theta$ are
\begin{equation*}
[L^*(y),R^*(y)]= \left[k, k+w\right]\qmq{for} y_k\le y<y_{k+1}.
\end{equation*} The  recursion for computing the $y_k$ is \eqref{gen.recur.disc} with $F_k$ the c.d.f.\ of $Y$ with parameter $\theta=k$ and the initial values in \eqref{gen.y.init} are
\begin{equation*}
 y_{-w}=y_{-w+1}=\ldots =  y_0=-1/2.
\end{equation*} The formulas \eqref{bin.Y.cdf}-\eqref{bin.Y.quant} given for the binomial hold for the hypergeometric when the $G_k,g_k$ there are replaced by the c.d.f.\ and density, respectively, of $X$ with $\theta=k$.

\subsection{Symmetric intervals} \label{sec:symm.hyper}
For obtaining symmetric intervals for the hypergeometric we make a similar recommendation for their modification as with the binomial in Section~\ref{sec:binom.symm}. Intervals~$[L(Y),R(Y)]$ for the hypergeometric parameter~$\theta$ are  symmetric if 
\begin{equation*}
[L(n-Y),R(n-Y)]=[N-R(Y),N-L(Y)],
\end{equation*}
which the  Push hypergeometric intervals in general  do not satisfy. Hypergeometric intervals~$[L(Y),R(Y)]$ can be replaced by the symmetric
\begin{multline}
\label{UWIM.hyper}
[L_{sym}(Y), R_{sym}(Y)]=\\
 [L(Y) \wedge (N - R(n - Y)), R(Y) \vee (N - L(n-Y))],
\end{multline}
which contains $[L(Y),R(Y)]$. For users requiring symmetric, $\gamma$ confidence intervals we suggest the following:
\begin{enumerate}
\item Find the smallest width $w^*\in\{1,\ldots, N\}$ for which the Push intervals~$[L^*(Y), R^*(Y)]$ exist for confidence level $\gamma$.
\item Obtain the symmetric intervals $[L_{sym}^*(Y), R_{sym}^*(Y)]$ given by \eqref{UWIM.hyper}.
\end{enumerate}
The resulting symmetric intervals~$[L_{sym}^*(Y), R_{sym}^*(Y)]$ will have confidence level $\gamma$ but may be wider than the optimal width~$w^*$  given by Theorem~\ref{thm:main} for the hypergeometric. We investigate the achieved widths of these intervals in our numerical simulations in the next section.

\subsection{Simulation examples and comparisons}\label{sec:hyper.sim}
In this section we compare the hypergeometric Push intervals, and their symmetric version, to simple fixed-width~$w$ intervals for the hypergeometric.  For comparison, we use the centered fixed-width interval \(XN/n\pm w/2\), analogous to the baseline considered for the binomial.

We consider sample sizes $n = 10$ and $n = 20$ with the population size fixed  at $N = 500$. Figures~\ref{fig:hyper_avg_results} and \ref{fig:hyper_avg_results_sym} contain the coverage probability of the Push intervals, their $[0,N]$-constrained and symmetric modifications, and the standard intervals, as functions of $\theta$. Figure~\ref{fig:hyper_min_results} shows the minimum coverage probability of these intervals as a function of their widths for $n = 10$.

The data for these figures was generated as follows. First, the minimum width  $w^*$ for which the Push intervals exist was computed for $\gamma=.9, .95$. Then the coverage probability of the intervals was estimated using $2,000$ realizations of $X$  for each $\theta = \theta_k=0,1,\ldots, N$.

\begin{figure}[b]
  \centering
  \includegraphics[width=9cm]{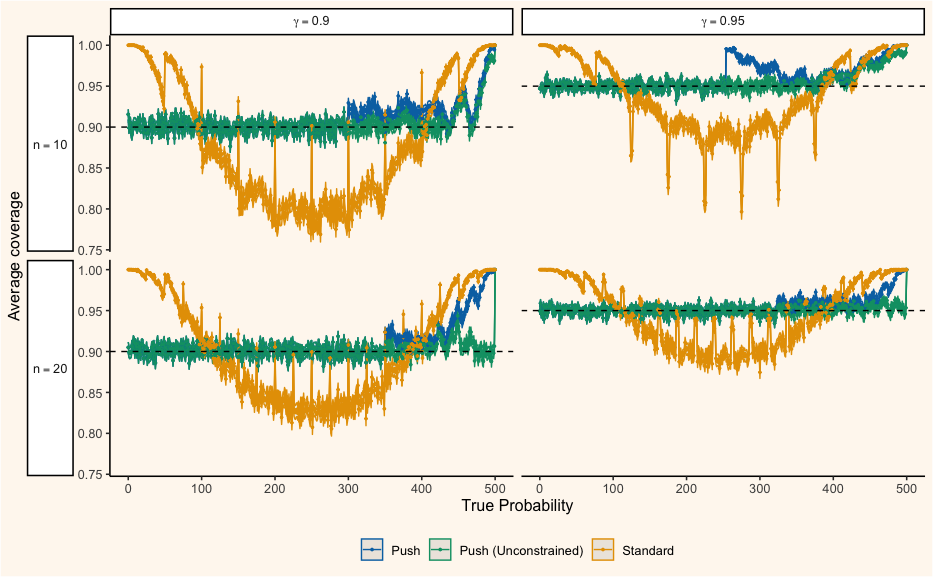}
  \caption{Coverage probability of intervals for hypergeometric~$\theta$ with sample sizes $n =10, 20$, population size~$N=500$, and nominal confidence levels~$\gamma=.9, .95$.}
  \label{fig:hyper_avg_results}
\end{figure}

\begin{figure}[b]
  \centering
  \includegraphics[width=9cm]{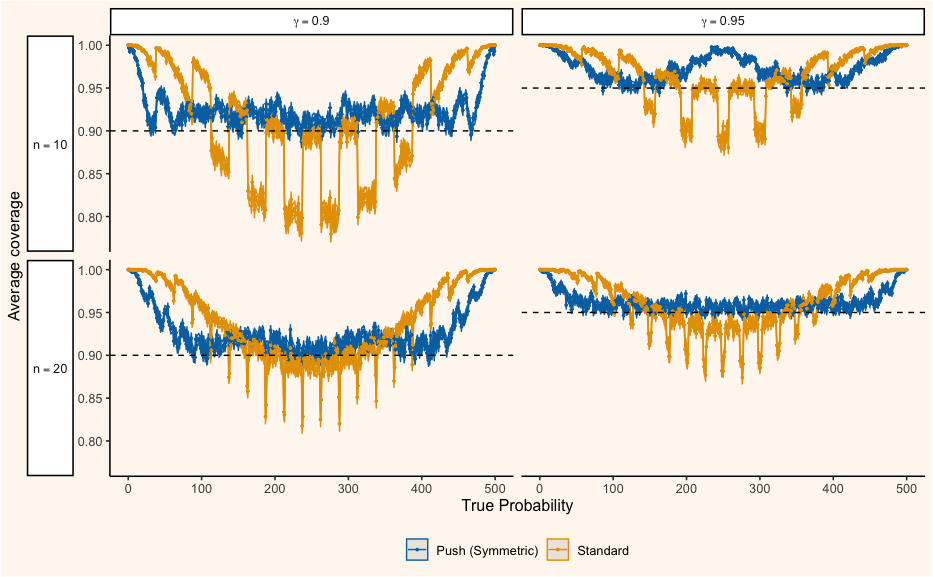}
  \caption{Coverage probability of symmetric intervals for hypergeometric~$\theta$ with sample sizes $n =10, 20$, population size~$N=500$, and nominal confidence levels~$\gamma=.9, .95$.}
  \label{fig:hyper_avg_results_sym}
\end{figure}

\begin{figure}[b]
  \centering
  \includegraphics[width=9cm]{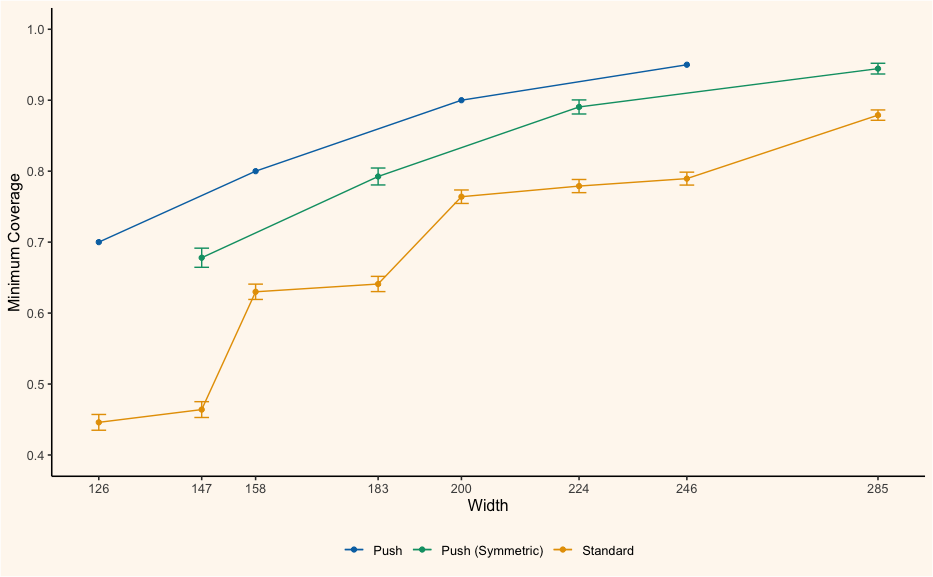}
  \caption{Minimum coverage probability over $\theta=0,1,\ldots, N=500$ of hypergeometric Push and standard intervals as a function of their achieved width.}
  \label{fig:hyper_min_results}
\end{figure}

In Figures~\ref{fig:hyper_avg_results} and \ref{fig:hyper_avg_results_sym} the coverage probability of the standard intervals falls well below the nominal level~$\gamma$ achieved by the Push intervals for $\theta$ near the center of $[0,N]$, and approaches $1$ for $\theta$ near the endpoints of this interval. The coverage probability of the unconstrained Push intervals in Figure~\ref{fig:hyper_avg_results} is centered at or just above $\gamma$, except for $\theta$ near the population size~$N$ where it rises due to the spillover of the right endpoints greater than $N$. The coverage probability of the $[0,N]$-constrained Push intervals in Figure~\ref{fig:hyper_avg_results} is nearly indistinguishable until $\theta$ near $N$ where the constraint causes the Push coverage to rise above $\gamma$.

Figure~\ref{fig:hyper_avg_results_sym} shows that the favorable performance of the Push intervals in Figure~\ref{fig:hyper_avg_results} is not just due to the method's asymmetry. Figure~\ref{fig:hyper_avg_results_sym} shows the same coverage probability of the standard intervals, which are symmetric by design, compared with the symmetric Push intervals as described in Section~\ref{sec:symm.hyper}. The symmetry constraint causes the coverage probability to rise for $\theta$ near $N$ and, in some cases, near the center of that interval $[0,N]$, but it remains no smaller than $\gamma$.

Figure \ref{fig:hyper_min_results} shows the minimum over $\theta$ of the three methods' coverage probabilities as a function of their widths. For Push these are the minimum widths~$w^*$ corresponding to confidence levels $\gamma=.7, .8, .9, .95$, and for the symmetric modification of Push these are the possibly slightly larger widths obtained by the steps in Section~\ref{sec:symm.hyper}.  For the standard intervals, the minimum coverage probability is shown at the union of these widths. Both Push and its symmetric modification show substantially higher minimum coverage probability, similar to the binomial case in Figure~\ref{fig:binom_min_results}.

\section{Confidence intervals for the normal mean}\label{sec:norm.mean}

A similar approach can be applied to obtain optimal fixed-width confidence intervals for the normal mean~$\theta$, based on observations with known variance, if one assumes a priori bounds  on $\theta$. Thus consider confidence intervals for $\theta$ based on $Y\sim N(\theta,\sigma^2)$, where $\sigma$ is known and $\theta$ is only known to lie in an interval $[\underline{\theta}, \overline{\theta}]$. An equivalent interpretation is that the coverage probability level is only required to hold for $\theta$ in this interval; see the citations in Section~\ref{sec:setup}.  Of course, $Y$ may represent the sample mean of i.i.d.\ normal observations with known variance, by appropriately adjusting $\sigma$ above.  Here, similar to the binomial case above, we discretize the parameter space with the grid
$$\theta_k = \underline{\theta}+(\overline{\theta}- \underline{\theta})k/m, \quad k=0,1,\ldots,m,$$ and the desired width is represented as $w=(\overline{\theta}- \underline{\theta})r/m$ for integer~$r$. Since $Y$ is already continuous, we do not impose additional randomization on the statistic. The Push intervals for $\theta$  are given by \eqref{gen.push.def} and the  recursion for computing the $y_k$ is \eqref{gen.recur.cont} with  $F_k(y)=\Phi((y-\theta_k)/\sigma)$ the c.d.f.\ of $Y$ under $\theta=\theta_k$, and $\Phi$ is the standard normal c.d.f. The initial values in \eqref{gen.y.init} are
\begin{equation*}
 y_{-r}=y_{-r+1}=\ldots =  y_0=-\infty.
\end{equation*} 

\subsection{Numerical comparison}
The standard intervals for this setting are the well-known $z$ intervals, whose fixed-width version has endpoints~$Y\pm w/2$, and whose coverage probability can be computed exactly as
\begin{equation}\label{z.min.cov}
P_\theta(|Y-\theta|\le w/2) = 1-2\Phi(-w/(2\sigma)).
\end{equation}
The coverage probability and width of this standard method and the Push intervals are calculated in Table~\ref{tab:normal_results} in the scenario $\underline{\theta} = -10$, $\overline{\theta} = 10$, $\sigma = 1$, and $m = 10^5$. The minimal Push widths~$w^*$ achieving $\gamma = .7, .8, .9, .95$ were computed and then coverage of the corresponding $z$ interval of width~$w = w^*$ was calculated using \eqref{z.min.cov}.   Since these coverage probabilities are less than~$\gamma$ throughout the table, in the last column we also report what larger width would be needed for the $z$ interval to achieve coverage probability~$\gamma$, calculated using the quantiles of \eqref{z.min.cov}.

\begin{table}
\caption{Minimum coverage probability (cov.\ prob.) and widths of the Push and standard $z$ intervals for the normal mean with bounds $[\underline{\theta},\overline{\theta}] = [-10,10]$  and known variance~$\sigma^2 = 1$.}
\tabcolsep=0pt
\begin{tabular*}{\columnwidth}{@{\extracolsep{\fill}}cccc@{}}
\hline
\makecell[c]{Push intervals\\min.\ cov.\ prob.~$\gamma$} &
\makecell[c]{$z$ intervals\\cov.\ prob.} &
\makecell[c]{Width\\$w^*$} &
\makecell[c]{$z$ interval width\\for cov.\ prob.\ $\geq \gamma$} \\
\hline
0.700 & 0.684 & 2.004 & 2.073\\
0.800 & 0.788 & 2.494 & 2.563\\
0.900 & 0.891 & 3.203 & 3.290\\
0.950 & 0.944 & 3.822 & 3.920\\
\hline
\label{tab:normal_results}
\end{tabular*}
\end{table}

In the table the Push intervals show small but consistent savings in width over the $z$ intervals in this scenario. This may be surprising due to the well-known optimality properties of $z$ intervals, however the bounds $\underline{\theta}, \overline{\theta}$ and asymmetry of the Push intervals are features that open the door to such improvement.

\section{Data analysis example}\label{sec:WHO.dat}

In this section we use the Push algorithm to analyze data from the Global Adult Tobacco Survey,  accessed via the World Health Organization (WHO) NCD Microdata Repository at \nolinkurl{https://extranet.who.int/ncdsmicrodata/index.php/home}. The WHO GATS Sample Design Manual \citeyearpar[GATS,][]{GATS20} specifies that $95\%$ confidence intervals for national-level estimates have margin of error no greater than 3 percentage points, in other words that the confidence intervals for $p$ have width no larger than $.06$. Thus, this is a natural setting to consider fixed-width confidence intervals. The 2016 study \citeyearpar[see TISS,][]{GATS18} was conducted in India  and interviewed $74,037$ participants aged 15 years and older about their use of various forms of nicotine. Details regarding the experimental design and other standard protocols can be found in the WHO's GATS Manual~\citeyearpar[GATS,][]{GATS20}.

As a target population we focus on respondents of age $18$ or older and on the possible usage of smoked tobacco; smokeless tobacco users were omitted. Ultimately, we retain $n = 56,026$ records with participants ranging in age from $18$ to $110$. We define four age group categories -- 18-24, 25-44, 45-64, and 65 and older --  and four education level categories: None (below high school level), high school diploma, undergraduate degree, and post-graduate degree. The response counts in the resulting  $16$ groups range from $124$ to $17,669$, and are given by the $n$ values in each of the $16$ cells of Table~\ref{tab:datagrid}.

To analyze this data with the Push and standard binomial intervals, we first calculate the smallest width for which the Push 95\% intervals (with $m = 10^5$) exist for the sample size in each age/education category. These widths and the 95\% Push interval for the proportion of adult tobacco smokers in India in that  category are  given in the second line of each table cell. The third line in each cell is the minimum width needed by the standard fixed-width binomial intervals~\eqref{bin.std.int} to maintain minimum coverage probability at least 95\% for that sample size~$n$, followed by the standard 95\% interval itself for proportion of adult tobacco smokers.

\begin{table*}
\caption{Push and standard 95\% confidence intervals (CIs) applied to 2016 Global Adult Tobacco Survey data: Each table cell has sample size $n$ (first row), minimal width and CI of Push (second row) and standard (third row) intervals for proportion of smoked tobacco users by education (rows) and age (columns) categories.}
\label{tab:datagrid}
\begin{tabular}{@{}lcccc@{}}
\hline
 & \textbf{18--24} & \textbf{25--44} & \textbf{45--64} & \textbf{65+} \\
\hline
\textbf{None} & \makecell[c]{$n=667$ \\ .073 [.027, .100] \\ .076 [.000, .076]} & \makecell[c]{$n=5,482$ \\ .026 [.106, .132] \\ .027 [.100, .126]} & \makecell[c]{$n=5,078$ \\ .027 [.172, .199] \\ .028 [.167, .195]} & \makecell[c]{$n=2,269$ \\ .040 [.174, .214] \\ .041 [.167, .208]} \\
\hline
\textbf{High School} & \makecell[c]{$n=6,810$ \\ .023 [.040, .063] \\ .024 [.032, .056]} & \makecell[c]{$n=17,669$ \\ .015 [.114, .129] \\ .015 [.111, .126]} & \makecell[c]{$n=8,165$ \\ .021 [.172, .193] \\ .022 [.168, .190]} & \makecell[c]{$n=2,032$ \\ .043 [.158, .200] \\ .044 [.149, .193]} \\
\hline
\textbf{Undergraduate} & \makecell[c]{$n=1,158$ \\ .056 [.012, .068] \\ .058 [.000, .046]} & \makecell[c]{$n=3,116$ \\ .034 [.055, .089] \\ .035 [.044, .079]} & \makecell[c]{$n=1,033$ \\ .059 [.066, .126] \\ .061 [.049, .110]} & \makecell[c]{$n=229$ \\ .122 [.060, .182] \\ .131 [.022, .153]} \\
\hline
\textbf{Post-graduate} & \makecell[c]{$n=225$ \\ .123 [.002, .125] \\ .133 [.000, .076]} & \makecell[c]{$n=1,595$ \\ .048 [.026, .074] \\ .050 [.008, .058]} & \makecell[c]{$n=554$ \\ .080 [.046, .126] \\ .085 [.019, .104]} & \makecell[c]{$n=124$ \\ .162 [.018, .180] \\ .177 [.000, .129]} \\
\hline
\end{tabular}
\end{table*}

The Push intervals show a small but consistent savings in interval width throughout the table.  The savings is most pronounced in cells with relatively smaller sample sizes, such as for post-graduate respondents age 65+ where Push provides a savings in maximum width of more than 8\%.  On the other hand, for high school respondents ages 25-44 with the largest sample size $n=17,669$, the difference in widths is smaller than the 3 decimal places reported in the table.

Table~\ref{tab:datagrid} shows that the maximum width~$.06$ prescribed  by the WHO GATS Sample Design Manual \citeyearpar[GATS,][]{GATS20} is attained in some categories but not others, and in categories such as education ``None'' (below high school) and ages 18-24 where the Push width exceeds $.06$, there is no fixed-width 95\% interval that can achieve that width for the current sample size, by Theorem~\ref{thm:main}.  An interesting case is respondents with undergraduate education ages $45$-$64$, where the Push width of $.059$ achieves the WHO specification but the smallest standard interval width of $.061$ does not.

\section{Conclusions and discussion}\label{sec:Concl}
We have proposed the Push method for fixed-width confidence intervals for a single, bounded parameter, extending a method for the binomial due to Asparaouhov and Lorden to a wider class of distributions including the hypergeometric and the normal mean with known variance. 
Under the remaining regularity conditions of Theorem~\ref{thm:main}, TP$_3$ is sufficient for the quasiconcavity condition and the optimality of the method in the theorem applies to distributions with a continuous, bounded parameter satisfying these conditions, including one-parameter exponential families  with monotonic sufficient statistic and natural parameter; in other settings, authors \citep[e.g.,][]{Choi17,Post93} have noted this unimodality-preserving property of TP$_3$ for interval probabilities. One example is the normal mean with known variance considered in Section~\ref{sec:norm.mean} where the method provides small but consistent savings over the venerable $z$ intervals.

 The failure of the binomial intervals~\eqref{bin.std.int} seen in Section~\ref{sec:binom.sim} in comparison with the Push intervals is a well-known phenomenon  in the statistics literature \citep[e.g.,][]{Brown01,Brown02} for the similar,  non-fixed-width binomial intervals. The optimized adjusted intervals in Table~\ref{tab:binom_push_shrink} are able to close much of this gap to Push in terms of width. In analyzing  the WHO tobacco use data in Section~\ref{sec:WHO.dat}, the Push intervals provide most  savings relative to the standard binomial intervals in small or moderate sample sizes, and the difference decreases for large $n$.

Our focus here has been on the criterion of minimizing the worst-case maximum width of confidence intervals, but alternative approaches exist in the literature for comparing variable-width intervals, such as expected length, an important and well-established alternative objective; see \citet{Brown01} and \citet{Pratt61}.

\begin{appendix}
\section*{Appendix}\label{appn} 

\begin{proof}[Proof of Lemma~\ref{lem:grid.cov}]
Part~(i): Assume that \eqref{rest.conf.lev} holds and fix $k$.  For $\theta_{k-1}<\theta<\theta_k$ we have that 
\begin{multline*}
\gamma\le P_\theta(\theta\in [L(Y),R(Y)]) \\
= P_\theta(L(Y)\le \theta_{k-1}<\theta_k\le R(Y)).
\end{multline*}
This last event does not depend  on $\theta$, so \eqref{gen.grid.cov} follows by taking $\theta\searrow\theta_{k-1}$ and $\theta\nearrow\theta_k$ and using continuity of $P_\theta$ in $\theta$.

For the converse, assume that \eqref{gen.grid.cov} holds. Then \eqref{rest.conf.lev} holds for $\theta=\theta_k$ and $\theta_{k-1}$, so it remains to verify for $\theta_{k-1}<\theta<\theta_k$, for which
\begin{multline}
\label{gen.cov=grid}
P_\theta(\theta\in [L(Y),R(Y)]) \\
= P_\theta(L(Y)\le \theta_{k-1}<\theta_k\le R(Y))\\
=P_\theta(Y\in(a_k, b_k)),
\end{multline}
where $a_k=\inf\{y:\, R(y)\ge \theta_k\}$ and  $b_k=\inf\{y:\, L(y)\ge \theta_k\}$ do not depend on $\theta$.  Let $\pi(\theta)$ denote \eqref{gen.cov=grid}. By quasiconcavity, the set 
\begin{equation*}
\mT=\{\theta\in[\theta_{k-1},\theta_k]: \pi(\theta)\ge \pi(\theta_{k-1})\wedge \pi(\theta_k)\}
\end{equation*} is an interval.  Since $\theta_{k-1}, \theta_k\in \mT$ by definition, it must be that $\mT=[\theta_{k-1},\theta_k]$, i.e.,  $\pi(\theta)\ge \pi(\theta_{k-1})\wedge \pi(\theta_k)\ge \gamma$ for all $\theta\in[\theta_{k-1},\theta_k]$, completing the proof of this part.

Part~(ii): Assume that $(\theta,y)\mapsto f_\theta(y)$ is TP$_3$.   Fix an interval $I\subseteq \mY$ and let $\pi(\theta)=P_\theta(Y\in I)$. We will show that, for an arbitrary value $c$, the set  $$\mT=\{\theta\in[\underline{\theta},\overline{\theta}]: \pi(\theta)\ge c\}$$ is an interval. Alternatively, write $$\mT=\{\theta\in[\underline{\theta},\overline{\theta}]: g(\theta)\ge 0\}$$ where $g(\theta)=\pi(\theta)-c=E_\theta [h(Y)]$ and $h(y)=\bm{1}\{y\in I\}-c$. Thus, to show that $\mT$ is an interval it suffices to show that $g(\theta)$ has at most $2$ sign changes in $[\underline{\theta},\overline{\theta}]$ and, if it has $2$ sign changes then its sign pattern must be $-+-$.  By comparing the endpoints of $I$ to those of $\mY$, it is apparent that $h(y)$ has at most $2$ sign changes in $\mY$ and, if $2$ sign changes then its sign pattern must be $-+-$. 
By  \citet[][p.~21]{Karlin68}, the variation-diminishing property of TP\(_3\)   implies that \(g\) has at most as many sign changes as \(h\). Moreover, if \(g\) has the same number of sign changes as \(h\), the two functions exhibit the same sequence of signs as their respective arguments increase. Hence \(g\) has at most two sign changes, and if it has $2$ their sign pattern is $-+-$.
\end{proof}

\end{appendix}
%
%
%

\begin{acks}[Acknowledgments]
The authors thank the Editor and an anonymous referee for their careful reading and constructive comments, particularly for identifying a significant gap in the theory of an earlier version of the paper.
\end{acks}

\def\cprime{$'$}

\end{document}